\documentclass[twocolumn,journal,twoside]{IEEEtran}

\usepackage[utf8]{inputenc} 
\usepackage[T1]{fontenc}    
\usepackage[bookmarks=false,draft]{hyperref}       
\usepackage{url}            
\usepackage{booktabs}       
\usepackage{amsfonts}       
\usepackage{nicefrac}       
\usepackage{microtype}      

\usepackage[cmex10]{amsmath}
\usepackage[font={small}]{caption}

\usepackage{graphicx}
\usepackage{pstricks}
\usepackage{color}
\usepackage{tcolorbox}

\usepackage{amssymb}
\usepackage{amsthm}

\usepackage{relsize}
\usepackage{bbm}
\usepackage{bm}
\usepackage{cite}

\usepackage{subcaption}

\makeatletter
\def\thm@space@setup{\thm@preskip=2pt
	\thm@postskip=2pt \itshape}
\makeatother

\newcommand{\Qian}[1]{{#1}}

\newtheoremstyle{newstyle}      
{} 
{} 
{\mdseries} 
{} 
{\bfseries} 
{.} 
{ } 
{} 

\theoremstyle{newstyle}

\newtheorem{theorem}{Theorem}
\newtheorem{lemma}{Lemma}

\newtheorem{corollary}{Corollary}

\theoremstyle{definition}

\newtheorem{definition}{Definition}

\theoremstyle{remark}
\newtheorem{remark}{Remark}

	\title{\huge{Straggler Mitigation	in Distributed Matrix Multiplication: Fundamental Limits and Optimal Coding}}
		\author{Qian~Yu$^{*}$, Mohammad~Ali~Maddah-Ali$^{\dagger}$, and A.~Salman~Avestimehr$^{*}$\\
		$^{*}$ Department of Electrical Engineering, University of Southern California, Los Angeles, CA, USA \\ 
		$^{\dagger}$ Nokia Bell Labs, Holmdel, NJ, USA\\
		\thanks{Manuscript received January 23, 2018; revised May 16, 2019; accepted December 12, 2019. This material is based upon work supported by Defense Advanced Research Projects Agency (DARPA) under Contract No. HR001117C0053. The views, opinions, and/or findings expressed are those of the author(s) and should not be interpreted as representing the official views or policies of the Department of Defense or the U.S. Government. This work is also in part supported by ONR award N000141612189 and NSF Grants CCF-1703575 and NeTS-1419632. A shorter version of this paper was presented at ISIT 2018 \cite{8437563}.}
			\thanks{Q.~Yu and A.S.~Avestimehr are with the Department of Electrical Engineering, University of Southern California, Los Angeles, CA, 90089, USA (e-mail:  qyu880@usc.edu; avestimehr@ee.usc.edu).}
\thanks{M. A. Maddah-Ali is with Nokia Bell Labs (e-mail: Mohammad.maddah-ali@nokia-bell-labs.com).}
\thanks{
Communicated by K. Narayanan, Associate Editor for Coding Techniques. }
\thanks{Copyright (c) 2020 IEEE. Personal use is permitted, but republication/redistribution requires IEEE permission.}
	}

\begin{document}

\maketitle

\begin{abstract}
We consider the problem of massive matrix multiplication, which underlies many data analytic applications, in a large-scale distributed system comprising a group of worker nodes.
We target the stragglers' delay performance bottleneck, which is due to the unpredictable latency in waiting for slowest nodes (or stragglers) to finish their tasks.
We propose a novel coding strategy, named \emph{entangled polynomial code}, for designing the intermediate computations at the worker nodes in order to minimize the recovery threshold (i.e., the number of workers that we need to wait for in order to compute the final output).
We demonstrate the optimality of entangled polynomial code in several cases, and show that it provides orderwise improvement over the conventional schemes for straggler mitigation.
Furthermore, we characterize the optimal recovery threshold among all linear coding strategies within a factor of $2$ using \emph{bilinear complexity}, by developing an improved version of the entangled polynomial code. {In particular, while evaluating bilinear complexity is a well-known challenging problem, we show that optimal recovery threshold for linear coding strategies can be approximated within a factor of $2$ of this fundamental quantity.}
\Qian{On the other hand, the improved version of the entangled polynomial code enables further and orderwise reduction in the recovery threshold, compared to its basic version.  
}
Finally, we show that the techniques developed in this paper can also be extended to several other problems such as coded convolution and fault-tolerant computing, leading to tight characterizations.

\end{abstract}

	\section{Introduction}

	Matrix multiplication is one of the key operations underlying many data analytics applications in various fields such as machine learning, scientific computing, and graph processing. Many such applications require  processing terabytes or even petabytes of data, which needs massive computation and storage resources that cannot be provided by a single machine. Hence, deploying  matrix computation tasks on large-scale distributed systems has received wide interests \cite{Cannon:1969:CCI:905686,CPE:CPE4330060702,CPE:CPE250,Solomonik:2011:CPM:2033408.2033420}.
	
There is, however, a major performance bottleneck that arises as we scale out computations across many distributed nodes: stragglers' delay bottleneck, which is due to the unpredictable latency in waiting for slowest nodes (or stragglers) to finish their tasks~\cite{dean2013tail}. The conventional approach for mitigating straggler effects  involves injecting some form of ``computation redundancy" such as repetition (e.g.,~\cite{zaharia2008improving}). Interestingly, it has been shown recently that \emph{coding theoretic} concepts can also play a transformational role in this problem, by efficiently creating ``computational redundancy'' to mitigate the stragglers~\cite{lee2015speeding,dutta2016short,tandon2016gradient,NIPS2017_7027,li2016unified, 7901473, 8663353}.

			\begin{figure}[htbp]
			\centering
			\includegraphics[width=0.95\linewidth]{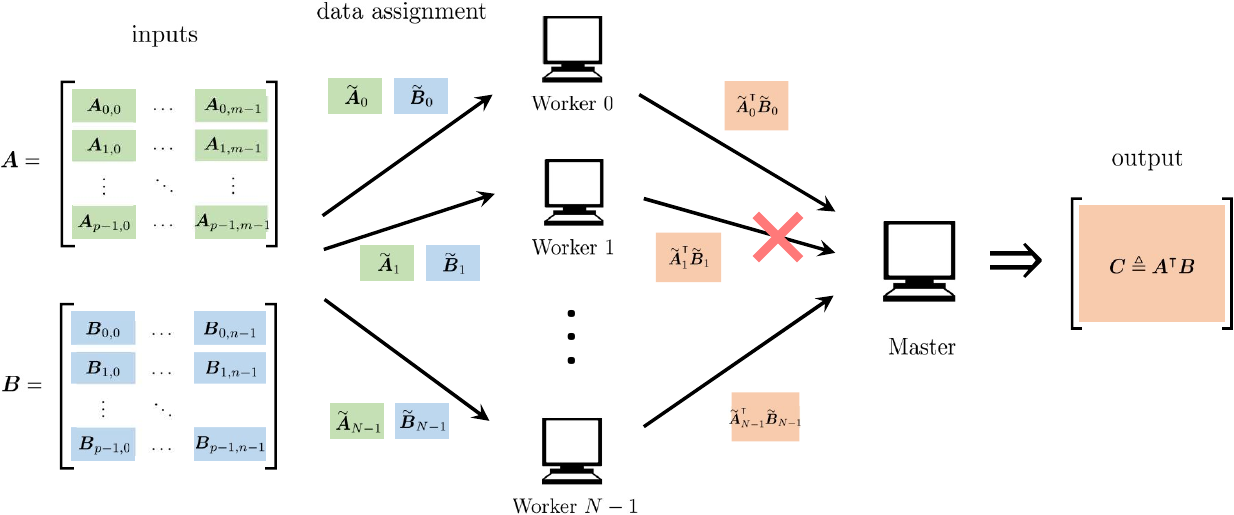}
			\caption{Overview of the distributed matrix multiplication problem. Each worker computes the product of the two stored encoded  submatrices ($\tilde{A}_i$ and $\tilde{B}_i$)
and returns the result to the master. By carefully designing the coding strategy, the master can decode the multiplication result of the input matrices from a subset of workers, without having to wait for stragglers (worker 1 in this example).
}
			\label{fig:sys}
		\end{figure}

In this paper, we consider a general formulation of distributed matrix multiplication, study information-theoretic limits, and develop optimal coding designs for straggler effect mitigation.  We consider a standard master-worker distributed setting, where a group of $N$ workers aim to collaboratively compute the product of two large matrices $A$ and $B$, and return the result $C=A ^\intercal B$ to the master. As shown in  Figure \ref{fig:sys}, the two input matrices are partitioned (arbitrarily) into $p$-by-$m$ and $p$-by-$n$ blocks of submatrices respectively, where all submatrices within the same input are of equal size. 
Each worker	has a local memory that can be used to store any coded function of each matrix, denoted by $ \tilde{A}_i$'s and $ \tilde{B}_i$'s, each with a size equal to that of the corresponding submatrices.  
	The workers then multiply their two stored (coded) submatrices and return the results to the master. 
	By carefully designing the coding functions, the master can decode the final result without having to wait for the slowest workers, which provides robustness against stragglers.


Note that by allowing different values of parameters $p$, $m$, and $n$, we allow flexible partitioning of input matrices, which in return enables different utilization of system resources (\Qian{e.g.}, the required amount of storage at each worker and the amount of communication from worker to master).\footnote{\Qian{A more detailed discussion is provided in Remark \ref{remark:trade}}} Hence, considering the system constraints on available storage and communication resources, one can choose $p$, $m$, and $n$ accordingly.  
We aim to find optimal coding and computation designs for \emph{any} choice of parameters $p$, $m$ and $n$, to provide optimum straggler effect mitigation for various situations.

With a careful design of the coded submatrices $\tilde{A}_i$ and $\tilde{B}_i$ at each worker, the master only needs results from the fastest workers before it can recover the final output, which effectively mitigates straggler issues. To measure the robustness against straggler effects of a given coding strategy, we use the metric \emph{recovery threshold}, defined previously in~\cite{NIPS2017_7027}, which is equal to the minimum number of workers that the master needs to wait for in order to compute the output $C$. Given this terminology, our main problem is as follows: What is the minimum possible recovery threshold and the corresponding coding scheme, for any choice of parameters $p$, $m$, $n$, and $N$?


We propose a novel coding technique, referred to as \emph{entangled polynomial code}, which achieves the recovery threshold of $pmn+p-1$ for all possible parameter values. 
The construction of the entangled polynomial code is based on the observation that when multiplying an $m$-by-$p$ matrix and a $p$-by-$n$ matrix, we essentially evaluate a subspace of bilinear functions, spanned by the pairwise product of the elements from the two matrices.  Although potentially there are a total of $p^2mn$ pairs of elements, at most $pmn$ pairs are directly related to the matrix product, which is an order of $p$ less.
The particular structure of the proposed code entangles the input matrices to the output such that the system almost avoids unnecessary multiplications and achieves a recovery threshold in the order of $pmn$, 
while allowing robust straggler mitigation for arbitrarily large systems. This allows orderwise improvement upon conventional uncoded approaches, random linear codes, and MDS-coding type approaches for straggler mitigation~\cite{lee2015speeding, dutta2016short}.


Entangled polynomial code generalizes our previously proposed polynomial code for distributed matrix multiplication~\cite{NIPS2017_7027}, which was designed for the special case of $p=1$ (i.e., allowing only column-wise partitioning of matrices $A$ and $B$). 
However, as we move to arbitrary partitioning of the input matrices (i.e., arbitrary values of $m$, $n$, and $p$), a key challenge is to design the coding strategy at each worker such that its computation best \emph{aligns} with the final computation $C$. 
In particular, to recover the product $C$, the master needs $mn$ components that each involve summing $p$ products of submatrices of $A$ and $B$. 
Entangled polynomial code effectively aligns the workers' computations with the master's need, which is its key distinguishing feature from polynomial code. 

We show that entangled polynomial code achieves the optimal recovery threshold among all linear coding strategies in the cases of $m=1$ or $n=1$. It also achieves the optimal recovery threshold among all possible schemes within a factor of $2$ when $m=1$ or $n=1$.

Furthermore, for \emph{all} partitionings of input matrices (i.e., all values of $p$, $m$, $n$, and $N$),  we characterize the optimal recovery threshold among all linear coding strategies  within a factor of $2$ of $R(p,m,n)$, which denotes the \emph{bilinear complexity} of multiplying an $m$-by-$p$ matrix to a $p$-by-$n$ matrix (see Definition~\ref{def:bi} later in the paper). While evaluating bilinear complexity is a well-known challenging problem in the computer science literature (see~\cite{gs005}), we show that the optimal recovery threshold for linear coding strategies can be approximated within a factor of $2$ of this fundamental quantity.

We establish this result by developing an {improved version of the entangled polynomial code}, 
which achieves a recovery threshold of $2R(p,m,n)-1$. Specifically, this coding construction exploits the fact that 
any matrix multiplication problem can be converted into a problem of computing the element-wise product of two arrays of length $R(p,m,n)$. Then we show that this augmented computing task can be optimally handled using a variation of the entangled polynomial code, and the corresponding optimal code achieves the recovery threshold $2R(p,m,n)-1$.      

Finally, we show that the coding construction and converse bounding techniques developed for proving the above results can also be directly extended to several other problems.
For example, we show that the converse bounding technique can be extended to the problem of coded convolution, which was originally considered in \cite{8006960}. We prove that the state-of-the-art scheme we proposed in \cite{NIPS2017_7027} for this problem is in fact optimal among all linear coding schemes. These techniques can also be applied in the context of fault-tolerant computing, which was first studied in \cite{FTC} for matrix multiplication. We provide tight characterizations on the maximum number of detectable or correctable errors.

We note that recently, another computation design named PolyDot was also proposed for distributed matrix multiplication, achieving a recovery threshold of $m^2(2p-1)$ for $m=n$ \Qian{\cite{8262882}}. Both entangled polynomial code and PolyDot are developed by extending the polynomial codes proposed in \cite{NIPS2017_7027}  to allow arbitrary partitioning of input matrices. Compared with PolyDot, entangled polynomial code achieves a {strictly} smaller recovery threshold of $pmn+p-1$, by a factor of $2$. More importantly,  in this paper we have developed a converse bounding technique that proves the optimality of the entangled polynomial code in several cases. We have also proposed an improved version of the entangled polynomial code and characterized the optimum recovery threshold within a factor of 2 for all parameter values.
	\label{sec:intro}

	\section{System Model and Problem Formulation}\label{sec:sys}

		We consider a problem of matrix multiplication with two input matrices $ A \in \mathbb{F}^{s\times r}$ and $ B \in \mathbb{F}^{s\times t}$, for some integers $r$, $s$, $t$ and a sufficiently large field $\mathbb{F}$.\footnote{Here we consider the general class of fields, which includes finite fields, the field of real numbers, and the field of complex numbers.} We are interested in computing the product  $ C \triangleq A ^\intercal B $ in a distributed computing environment with a master node and $N$ worker nodes, where
		each worker can store $\frac{1}{pm}$ fraction of $A$ and $\frac{1}{pn}$ fraction of $B$, based on some integer parameters $p$, $m$, and $n$ (see Fig. \ref{fig:sys}). 

		Specifically, each worker $i$ can store two coded matrices $\tilde{A}_i\in\mathbb{F}^{\frac{s}{p}\times \frac{r}{m}}$ and $\tilde{B}_i\in\mathbb{F}^{\frac{s}{p}\times \frac{t}{n}}$, computed based on 
		$ A $ and $ B $ respectively. 
		Each worker can compute the product  $\tilde{C}_i\triangleq\tilde{A}_i^\intercal\tilde{B}_i$, and return it to the master. The master waits only for the results from a subset of workers before proceeding to recover the final output $ C $ using certain \emph{decoding functions}.

    	Given the above system model, we formulate the \emph{distributed matrix multiplication problem} based on the following terminology: We define the \emph{computation strategy} as a collection of $2N$ \emph{encoding functions}, denoted by
    	\begin{align}
    	    {\boldsymbol{f}}=(f_0,{f}_1,...,{f}_{N-1}), \ \ \ \ \ \ \ \  {\boldsymbol{g}}=(g_0,{g}_1,...,{g}_{N-1}),
    	\end{align}
    	 that are used by the workers to compute 
    	each $\tilde{A}_i$ and $\tilde{B}_i$,
    	and a class of \emph{decoding functions},  denoted by
    	\begin{align}
    	    {\boldsymbol{d}}=\{d_{\mathcal{{K}}}\}_{\mathcal{{K}}\subseteq\{0,1,...,N-1\}}, 
    	\end{align}
    	that are used by the master to recover $C$ given results from any subset $\mathcal{{K}}$ of the workers. 
    Each worker $i$ stores matrices 
    	\begin{align}
    	\tilde{A}_i={f}_i( A ), \ \ \ \ \ \ \ \  \ \ \ \ \ \ 
	    \tilde{B}_i={g}_i( B ), 
    	\end{align}
    	and the master can compute an estimate $\hat{C}$ of matrix $C$ using results from a subset $\mathcal{K}$ of the workers by computing
    	\begin{align}
    	\hat{C}=d_{\mathcal{{K}}}\left(\{\tilde{C}_i\}_{i\in\mathcal{K}}\right).
    	\end{align}

    	For any integer $k$, we say a computation strategy is \emph{$k$-recoverable} if the master can recover $ C $ given the computing results from \emph{any} $k$ workers. Specifically, a computation strategy is \emph{$k$-recoverable} if for any subset $\mathcal{K}$ of $k$ users, the final output $\hat{C}$ from the master equals $C$ for all possible input values.
    	We define the \emph{recovery threshold} of a computation strategy, denoted by $K( 
    	{\boldsymbol{f}}
    	,{\boldsymbol{g}}, {\boldsymbol{d}})$, as the minimum integer $k$ such that computation strategy  $({\boldsymbol{f}},{\boldsymbol{g}}, {\boldsymbol{d}})$ is $k$-recoverable. 
    
    	We aim to find a computation strategy that requires the minimum possible recovery threshold and allows efficient decoding at the master. Among all possible computation strategies, we are particularly interested in a certain class of designs, referred to as the \emph{linear codes} and defined as follows:
	    	\begin{definition}
    	For a distributed matrix multiplication problem of computing $A ^\intercal B $ using $N$ workers, we say a computation strategy is a \emph{linear code} given parameters $p$, $m$, and $n$, if there is a partitioning of the input matrices $A$ and $B$ where each matrix is divided into the following submatrices of equal sizes
    		\begin{align}
	   A = &
 \begin{bmatrix}
  A_{0,0} & A_{0,1} & \cdots & A_{0,m-1} \\
  A_{1,0} & A_{1,1} & \cdots & A_{1,m-1} \\
  \vdots  & \vdots  & \ddots & \vdots  \\
  A_{p-1,0} & A_{p-1,1} & \cdots & A_{p-1,m-1} 
 \end{bmatrix}, \label{eq:a}\\
	    B = &
 \begin{bmatrix}
  B_{0,0} & B_{0,1} & \cdots & B_{0,n-1} \\
  B_{1,0} & B_{1,1} & \cdots & B_{1,n-1} \\
  \vdots  & \vdots  & \ddots & \vdots  \\
  B_{p-1,0} & B_{p-1,1} & \cdots & B_{p-1,n-1} 
 \end{bmatrix},\label{eq:b}
	\end{align}
		such that the encoding functions of each worker $i$ can be written as 
    		\begin{align}
    	&\tilde{A}_i=\sum_{j,k} A_{j,k} a_{ijk},  \ \ \ \ \ \ \ \ 
	    \tilde{B}_i=\sum_{j,k} B_{j,k} b_{ijk}, 
    	\end{align}
    	for some tensors $a$ and $b$, 
    	and the decoding function given each subset $\mathcal{K}$ can be written as\footnote{Here $\hat{C}_{j,k}$ denotes the master's estimate of the subblock of $C$ that corresponds to $\sum_{\ell}A_{\ell,j}B_{\ell,k} $.}
    	 	\begin{align}
        	&\hat{C}_{j,k}=\sum_{i\in\mathcal{K}}\tilde{C}_{i}c_{ijk},
		\end{align}
    	for some tensor $c$. For brevity, we denote the set of linear codes as $\mathcal{L}$.
    	\end{definition}
    	
       The major advantage of linear codes is that they guarantee that both the encoding and the decoding complexities of the scheme scale linearly with respect to the size of the input matrices. Furthermore, as we have proved in \cite{NIPS2017_7027}, linear codes are optimal for $p=1$.
    	  Given the above terminology, we define the following concept.
    		\begin{definition}
    	For a distributed matrix multiplication problem of computing $A ^\intercal B $ using $N$ workers, we define the \emph{optimum linear recovery threshold}  as a function of the problem parameters $p$, $m$, $n$, and $N$, denoted by $K^*_{\textup{linear}}$, as the minimum achievable recovery threshold among all linear codes. 
    	Specifically,
    	\begin{align}
    	    K^*_{\textup{linear}}\triangleq  \min_{({\boldsymbol{f}},{\boldsymbol{g}}, {\boldsymbol{d}})\in\mathcal{L}} K({\boldsymbol{f}},{\boldsymbol{g}}, {\boldsymbol{d}}).
    	\end{align}
    	\end{definition}
    	
 
     	Our goal is to characterize the optimum linear recovery threshold $K^*_{\textup{linear}}$, and to find  computation strategies to achieve such optimum threshold. Note that if the number of workers $N$ is too small, obviously no valid computation strategy exists even without requiring straggler tolerance. Hence, in the rest of the paper, we only consider the meaningful case where $N$ is large enough to support at least one valid computation strategy. 
     	\Qian{More concretely, we show that the minimum possible number of workers is given by a fundamental quantity: the bilinear complexity of multiplying an $m$-by-$p$ matrix and a $p$-by-$n$ matrix, which is formally introduced in Section \ref{sec:res}}.
     	

    We are also interested in characterizing the minimum recovery threshold achievable using general coding strategies (including non-linear codes). Similar to \cite{NIPS2017_7027}, we define this value as the \emph{optimum recovery threshold} and denote it by $K^*$. 

    
    




	
	\section{Main Results}\label{sec:res}

	We state our main results in the following theorems:
	\begin{theorem}\label{th:scheme}
		For a distributed matrix multiplication problem of computing $A ^\intercal B $ using $N$ workers, with parameters $p$, $m$, and $n$, the following recovery threshold can be achieved by a linear code, referred to as the \emph{entangled polynomial code}.\footnote{For $N< pmn+p-1$, we define $K_{\textup{entangled-poly}}\triangleq N$.  }
		\begin{align}
		 K_{\textup{entangled-poly}}\triangleq pmn+p-1.
		\end{align}
	\end{theorem}

	\begin{remark}
	\label{remark:challenge}
	Compared to some other possible approaches, our proposed entangled polynomial code provides orderwise improvement in the recovery threshold (see Fig. \ref{fig:comp}). One conventional approach (referred to as the \emph{uncoded repetition scheme}) is to let each worker store and multiply uncoded submatrices. %
	With the additional computation redundancy through repetition, the scheme can robustly tolerate some stragglers. However, its recovery threshold, {$K_{\textup{uncoded}}\triangleq N-\lfloor\frac{N}{pmn}\rfloor+1$,} grows linearly with respect to the number of workers. Another approach is to let each worker store two random linear combinations of the input submatrices (referred to as the \emph{random linear code}). With high probability, this achieves recovery threshold $K_{\textup{RL}}\triangleq p^2mn$,\footnote{\Qian{Intuitively, because each worker returns a random linear combination of all $p^2mn$ possible pairwise products, with high probability, the final output can be recovered from any subset of $p^2mn$ results. }} which does not scale with $N$. However, to calculate $C$, we need the result of {at most} $pmn$ sub-matrix multiplications. Indeed, the lack of structure in the random coding forces the system to wait for $p$ times more than what is essentially needed. 
    One  surprising aspect of the proposed  entangled polynomial code is that, due to its particular structure which aligns the workers' computations with the master's need, it avoids  unnecessary multiplications of submatrices. As a result, it achieves a recovery threshold that does not scale with $N$, and is orderwise smaller than that of the random linear code. Furthermore,  it allows efficient decoding at the master, which requires at most an almost linear complexity.
	\end{remark}

			\begin{figure}[htbp]
			\vspace{-3mm}
			\centering
			\includegraphics[width=0.8\linewidth]{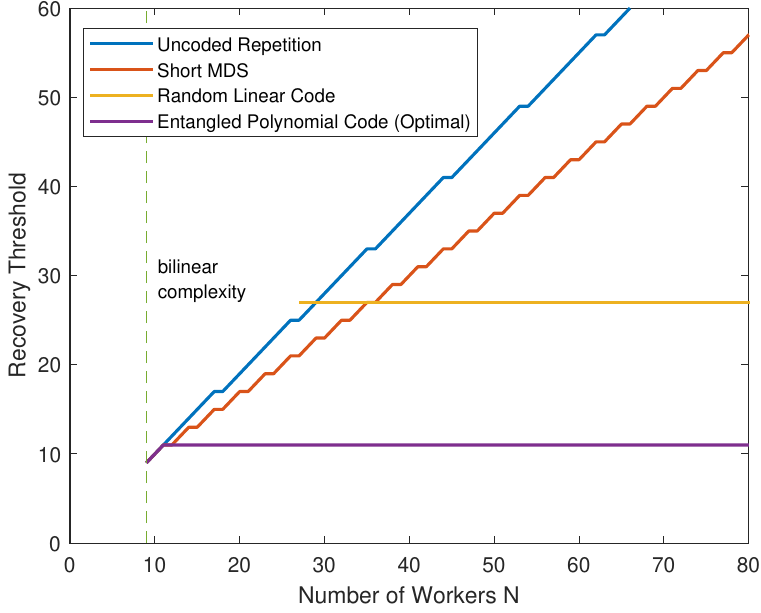}
			\caption{Comparison of the recovery thresholds achieved by the uncoded repetition scheme, the random linear code, the short-MDS (or short-dot) \cite{lee2015speeding, dutta2016short} and our proposed entangled polynomial code, given problem parameters $p=m=3$, $n=1$. 
			The entangled polynomial code orderwise improves upon all other approaches. It also achieves the optimum linear recovery threshold in this scenario.}
			\vspace{-3mm}
			\label{fig:comp}
		\end{figure}
	
	\begin{remark} 
    There have been several works in prior literature investigating the $p=1$ case \cite{lee2015speeding,high_mul,NIPS2017_7027}. For this special case, the entangled polynomial code reduces to our previously proposed polynomial code, which achieves the optimum recovery threshold $mn$ and orderwise improves upon other designs.   
    On the other hand, there has been some investigation on matrix-by-vector type multiplication \cite{lee2015speeding, dutta2016short}, which can be viewed as the special case of	$m=1$ or $n=1$ in our proposed problem. The short-MDS code (or short-dot) has been proposed, achieving a recovery threshold of $N-\lfloor \frac{N}{p}\rfloor+m$, which scales linearly with $N$. Our proposed entangled polynomial code also strictly and orderwise improves upon that (see Fig. \ref{fig:comp}).
    \end{remark}
	
	\begin{remark}\label{remark:trade}
\Qian{By selecting different values of parameters $p$, $m$, and $n$, the entangled polynomial code enables different utilization of the system resources, which allows for balancing the costs due to storage and communication.  
In particular, one can show that a distributed implementation for multiplying $ A^\intercal \in \mathbb{F}^{r\times s}$ and $ B \in \mathbb{F}^{s\times t}$ with parameters $p$, $m$, and $n$ requires:
	\begin{itemize}
	    \item Computation load at each worker (normalized by the cost of a single field operation):  $O(\frac{srt}{pmn})$,
	    \item Communication required from each worker (normalized by the size of $C$): $L\triangleq \frac{1}{mn}$,
	    \item Storage allocated for storing each coded matrix (normalized by the sizes of $A$, $B$, respectively): $\mu_A\triangleq \frac{1}{pm}$, $\mu_B\triangleq \frac{1}{pn}$. 
	   	\end{itemize}
If we roughly fix the computation load (specifically, fixing $pmn$ for the cubic matrix multiplication algorithm), the computing scheme requires the following trade-off between storage and communication:
\begin{align}
    L\mu_A\mu_B \sim \textup{constant}.
\end{align}
By designing the values of $p$, $m$, and $n$, we can operate at different locations on this trade-off to account for the system's requirement\footnote{
\Qian{For example, letting $p=1$ minimizes the communication load $L$, and letting $n=1$ or $m=1$ minimizes the storage cost for storing matrix $A$ or matrix $B$, respectively. Our proposed entangled polynomial code achieves the optimum linear recovery threshold in all these cases. More generally, adjusting the value of $p$ trades communication by storage; then adjusting the ratio between $m$ and $n$ allows for minimizing the overall storage cost, to account for the scenario where the sizes of input matrices are unbalanced.
Finally, by scaling $p$, $m$, and $n$ without taking the computational constraint into account, we enable the flexibility in terms of level of distribution. 
}}, while the entangled polynomial code maintains almost the same recovery threshold.}

	\end{remark}

    Our second result is the optimality of the entangled polynomial code when $m=1$ or $n=1$. Specifically, we prove that entangled polynomial code is optimal in this scenario among all linear codes. Furthermore, if the base field $\mathbb{F}$ is finite, it also achieves the optimum recovery threshold $K^*$ within a factor of $2$, with non-linear coding strategies taken into account.
	
		\begin{theorem}\label{th:li}
		For a distributed matrix multiplication problem of computing $A ^\intercal B $ using $N$ workers,  with parameters $p$, $m$, and $n$, if $m=1$ or $n=1$, we have
		\begin{align}\label{eq:tight}
		 K^*_{\textup{linear}}&= K_{\textup{entangled-poly}}.
		 \end{align}
		 Moreover, if the base field $\mathbb{F}$ is finite,
		 \begin{align}\label{bound:it}
		 \frac{1}{2}K_{\textup{entangled-poly}}<&K^*\leq K_{\textup{entangled-poly}}.
		 \end{align}
		 \end{theorem}
	
		\begin{remark}
		We prove Theorem \ref{th:li} by first exploiting the algebraic structure of matrix multiplication to develop a linear algebraic converse for equation (\ref{eq:tight}), and then constructing an information theoretic converse to prove inequality (\ref{bound:it}). The linear algebraic converse only relies on two properties of the matrix multiplication operation: 1) bilinearity, and 2) uniqueness of zero element. This technique can be extended to any other bilinear operations with similar properties, such as convolution, as mentioned later (see Theorem \ref{thm:conv}). On the other hand, the information theoretic converse is obtained through a cut-set type argument, which allows a lower bound on the recovery thresholds even for non-linear codes.


	
	\end{remark}


		Our final result on the main problem is characterizing the optimum linear recovery threshold $K^*_{\textup{linear}}$ within a factor of $2$ for \emph{all} possible  $p$, $m$, $n$, and $N$, by developing an {improved version of the entangled polynomial code}. This characterization involves the fundamental concept of bilinear complexity \cite{gs005}:

	  
	
		
	\begin{definition}\label{def:bi}
	The \emph{bilinear complexity} of multiplying an $m$-by-$p$ matrix and a $p$-by-$n$ matrix, denoted by $R(p,m,n)$, \Qian{is defined as} the minimum number of element-wise multiplications required to complete such an operation. Rigorously, $R(p,m,n)$ denotes the minimum integer $R$, such that we can find tensors $a\in\mathbb{F}^{R\times p\times m}$, $b\in \mathbb{F}^{R\times p\times n}$, and $c\in \mathbb{F}^{R\times m\times n}$, satisfying
	\begin{align}\label{eq:bidef}
    	\sum_{i}c_{ijk}\left(\sum_{j',k'} A_{j'k'} a_{ij'k'}\right) 
	    &\left(\sum_{j'',k''} B_{j''k''} b_{ij''k''}\right)\nonumber\\
	    &
	    = \sum_{\ell} A_{\ell j} B_{\ell k}.
	\end{align}
	for any input matrices $A\in\mathbb{F}^{p\times m}$, $B\in \mathbb{F}^{p\times n}$. 
	\end{definition}

    Using this concept, we state our result as follows. 
            \begin{theorem}\label{th:2}
            For a distributed matrix multiplication problem of computing $A ^\intercal B $ using $N$ workers, with parameters $p$, $m$, and $n$, the optimum linear recovery threshold is characterized by 
          \begin{align}\label{bounds:f2}
             R(p,m,n) \leq K^*_{\textup{linear}}\leq 2R(p,m,n)-1,
          \end{align}
          where $R(p,m,n)$ denotes the bilinear complexity of multiplying an $m$-by-$p$ matrix and a $p$-by-$n$ matrix.
      \end{theorem}
      
      
           \begin{remark}
    \Qian{The key proof idea of Theorem \ref{th:2} is twofold. We first demonstrate a one-to-one correspondence between linear computation strategies and upper bound constructions\footnote{Formally defined in Section \ref{sec:2}.} for bilinear complexity, which enables converting a matrix multiplication problem into computing the element-wise product of two vectors of length $R(p,m,n)$. Then we show that an optimal computation strategy can be developed for this augmented problem, which achieves the stated recovery threshold.  Similarly to this result, factor-of-$2$ characterization can also be obtained for non-linear codes, as discussed in Section \ref{sec:2}.   }
      \end{remark}
      
             \begin{remark}
          \Qian{   The coding construction we developed for proving Theorem \ref{th:2} provides an improved version of the entangled polynomial code. Explicitly, given any upper bound construction for $R(p,m,n)$ with rank $R$, the coding scheme achieves a recovery threshold of $2R-1$, while tolerating arbitrarily many stragglers. This improved version further and orderwise reduces the needed recovery threshold on top of its basic version. 
          For example, by simply applying the well-know Strassen's construction \cite{Strassen1969}, which provides an upper bound $R(2^k,2^k,2^k)\leq7^k $ for any $k\in\mathbb{N}$, the proposed coding scheme achieves a recovery threshold of $2\cdot 7^k-1$, which orderwise improves upon $K_{\textup{entangled-poly}}=8^k+2^k-1$ achieved by the entangled polynomial code. Further improvements can be achieved by applying constructions with lower ranks, up to $2R(p,m,n)-1$.   }    
      \end{remark}

      
      \begin{remark}
      In parallel to this work, the Generalized PolyDot scheme was proposed in \cite{8437852} to extend the PolyDot construction \cite{8262882} to asymmetric matrix-vector multiplication. Generalized PolyDot can be applied to achieve the same recovery threshold of the entangled polynomial code for special case of $m=1$ or $n=1$. However, entangled polynomial codes achieve (unboundedly) better recovery thresholds for general values of $p$, $m$, and $n$.
      \end{remark}
      
      
The techniques we developed in this paper can also be extended to several other problems, such as coded convolution \cite{8006960} and fault-tolerant computing \cite{FTC,1457803}, leading to tight characterizations. For coded convolution, we present our result in the following theorem.

	\begin{theorem}
	    \label{thm:conv} 
		For the distributed convolution problem of computing $\boldsymbol{a}*\boldsymbol{b}$ using $N$ workers that can each store $\frac{1}{m}$ fraction of $\boldsymbol{a}$ and $\frac{1}{n}$ fraction of $\boldsymbol{b}$, the optimum recovery threshold that can be achieved using linear codes, denoted by $K^*_{\textup{conv-linear}}$ , is exactly characterized by the following equation
		\begin{align}
		    K^*_{\textup{conv-linear}}=K_{\textup{conv-poly}}\triangleq m+n-1.
		\end{align}
	\end{theorem}
	
	
		\begin{remark}
			Theorem \ref{thm:conv} is proved based on our previously developed coded computing scheme for convolution, which is a variation of the polynomial code \cite{NIPS2017_7027}. As mentioned before, we extend the proof idea of Theorem \ref{th:li} to prove the matching converse. This theorem proves the optimality of the computation scheme in \cite{NIPS2017_7027} among all computation strategies where the encoding functions are linear. For detailed problem formulation and proof, see Appendix \ref{app:conv}.  
	\end{remark}

   Our second extension is in the fault-tolerant computing setting, \Qian{which was first discussed in \cite{FTC} for matrix multiplication.
   Unlike the straggler effects we studied in this paper, fault tolerance considers scenarios where} arbitrary errors can be injected into the computation, and the master has no information about which subset of workers are returning errors. We show that the techniques we developed for straggler mitigation can also be applied in this setting to improve robustness against computing failures, and the optimality of any encoding function in terms of recovery threshold also preserves when applied in the fault-tolerant computing setting. As an example, we present the following theorem,  demonstrating this connection.

	
	
\begin{theorem}\label{thm:ftc}
	For a distributed matrix multiplication problem of computing $A ^\intercal B $ using $N$ workers, with parameters $p$, $m$, and $n$,  if $m=1$ or $n=1$, the entangled polynomial code can detect up to  
		\begin{align}
		 E^*_{\textup{detect}} ={N-K_{\textup{entangled-poly}}}
		\end{align}	
		errors, and correct up to 
		\begin{align}
		 E^*_{\textup{correct}}=\left\lfloor\frac{N-K_{\textup{entangled-poly}}}{2}\right\rfloor
		\end{align}
		errors.
		This can not be improved using any other linear encoding strategies. 
	\end{theorem}

	\begin{remark}
			The proof idea for Theorem \ref{thm:ftc} is to connect the straggler mitigation problem and the fault tolerance problem by extending the concept of Hamming distance to coded computing.  Specifically, we map the straggler mitigation problem to the problem of correcting erasure errors, and the fault tolerance problem to the problem of correcting arbitrary errors. The solution to these two communication problems are deeply connected by the Hamming distance, and we show that this result extends to coded computing (see Lemma \ref{lemma:hamming} in Appendix \ref{app:ftc}).  
			Since the concept of Hamming distance is not exclusively defined for linear codes, this connection also holds for arbitrary computation strategies. Furthermore, this approach can be easily extended to the hybrid settings where both stragglers and computing errors exist, and similar results can be proved. The detailed formulation and proof can be found in Appendix \ref{app:ftc}. 
	\end{remark}

	In Section \ref{sec:scheme}, we prove Theorem \ref{th:scheme} by describing the entangled polynomial code. Then in Section \ref{sec:conv},  we prove Theorem \ref{th:li} by deriving the converses. Finally, we present the coding construction and converse for proving Theorem \ref{th:2} in Section \ref{sec:2}.

	\section{Entangled Polynomial Code}\label{sec:scheme}

	In this section, we prove Theorem \ref{th:scheme} by formally describing the entangled polynomial code and its decoding procedure. 
	We start with an illustrating example.

	\subsection{Illustrating Example}
	
	Consider a distributed matrix
multiplication task of computing $A ^\intercal B $ using $N=5$ workers that can each store half of the rows (i.e., $p=2$ and $m=n=1$). We evenly divide each input matrix along the row side into 2 submatrices:
		\begin{align}
	   A = 
 \begin{bmatrix}
  A_{0}  \\
  A_{1}  
 \end{bmatrix}, \ \ \ \ \ \ \ \ 
	    B = 
 \begin{bmatrix}
  B_{0} \\
  B_{1} 
 \end{bmatrix},
	\end{align}
Given this notation, we essentially want to compute 
 			\begin{align}
	  C=A ^\intercal B =  
 \begin{bmatrix}
 A _0^\intercal   B _0  + A _1^\intercal   B _1  
 \end{bmatrix}.
	\end{align}
	
	A naive computation strategy is to let the $5$ workers compute each $ A _i^\intercal   B _i$ uncodedly with repetition. Specifically we can let $3$ workers compute $ A _0^\intercal   B _0$ and $2$ workers compute $ A _1^\intercal   B _1$. However, this approach can only robustly tolerate $1$ straggler, achieving a recovery threshold of $4$.
	Another naive approach is to use random linear codes, i.e., let each worker store a random linear combination of $ A _0$, $ A _1 $, and a combination of $ B_0$, $ B_1 $. However, the resulting computation result of each worker is a random linear combination of $4$ variables $ A _0^\intercal   B _0$, $ A _0^\intercal   B _1$, $ A _1^\intercal   B _0$, and $ A _1^\intercal   B _1$, which also results in a recovery threshold of $4$.
	
	Surprisingly, there is a simple computation strategy for this example that achieves the optimum linear recovery threshold of $3$. The main idea is to instead inject structured redundancy tailored to the matrix multiplication operation. We present this proposed strategy as follows:

	 	\begin{figure}[htbp]
			\vspace{-3mm}
			\centering
 			\includegraphics[width=0.95\linewidth]{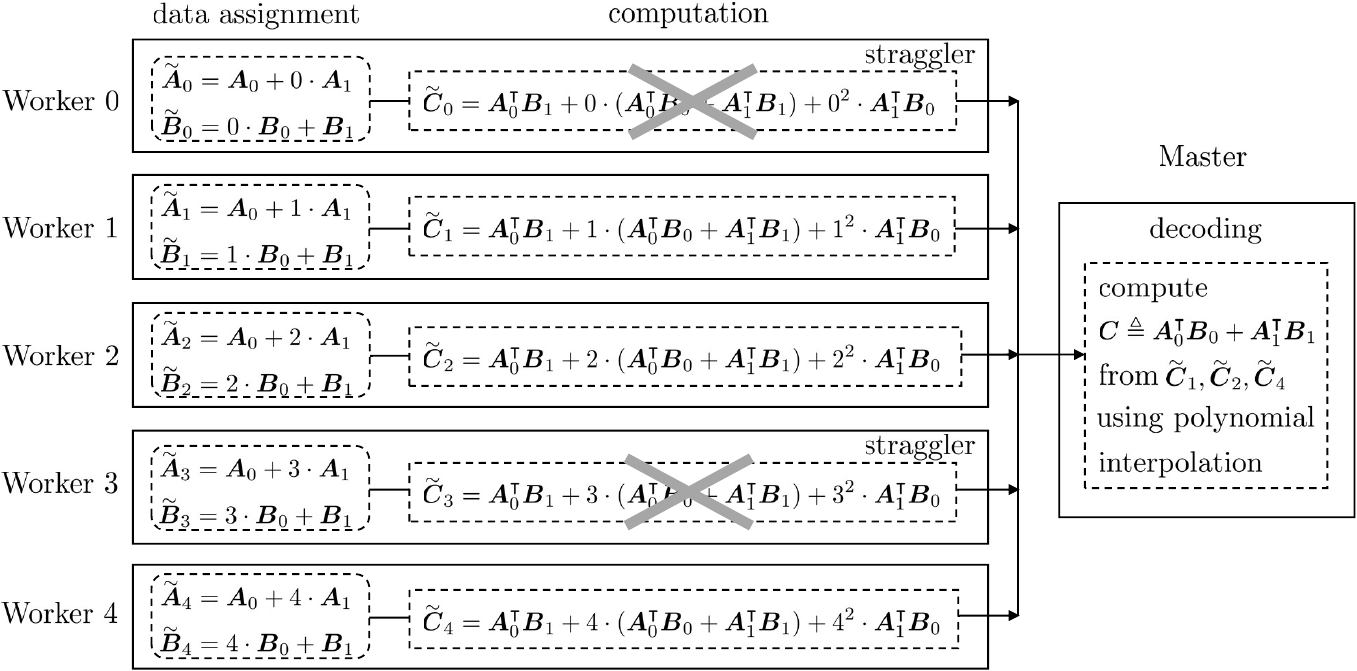}
			\caption{Example using entangled polynomial code, with $5$ workers that can each store half of each input matrix. (a) Computation strategy: each worker $i$ stores $A_0+iA_1$ and $iB_0+B_1$, and computes their product. (b) Decoding: master waits for results from \emph{any} $3$ workers, and decodes the output using polynomial interpolation. }
			\vspace{-3mm}
			\label{fig:exp}
		\end{figure}
	
	Suppose elements of $A, B$ are in $\mathbb{R}$. Let each worker $i\in \{0,1,...,4\}$ store the following two coded submatrices:
	\begin{align}
    	\tilde{A}_i= A _0+i A _1,  \ \ \ \ \ \ \ \ 
	    \tilde{B}_i= iB _0+ B _1. \label{eq:bexp}
    	\end{align}
    	To prove that this design gives a recovery threshold of $3$, we need to find a valid decoding function for any subset of $3$ workers.   	
We demonstrate this decodability through a representative scenario, where the master receives the computation results from workers $1$, $2$, and $4$, as shown in Figure \ref{fig:exp}. The decodability for the other $9$ possible scenarios can be proved similarly.
	
		 According to the designed computation strategy, we have  
 	\begin{align}\label{eq:fore}
 	\begin{bmatrix}
    \tilde{C}_1   \\   \tilde{C}_2  \\   \tilde{C}_4 
\end{bmatrix}
&=
	    \begin{bmatrix}
    1^0&1^1&1^2\\
   2^0&2^1&2^2 \\
   4^0&4^1&4^2 
\end{bmatrix}
 	\begin{bmatrix}
     A _0^\intercal   B _1   \\    A _0^\intercal   B _0 +
     A _1^\intercal   B _1    \\    A _1^\intercal   B _0 
\end{bmatrix}.
	\end{align}	
	The coefficient matrix in the above equation is a Vandermonde matrix, which is  invertible because its parameters $1,2, 4$ are distinct in $\mathbb{R}$. So one decoding approach is 
	to directly invert equation (\ref{eq:fore}), of which the returned result includes the needed matrix $C= A _0^\intercal   B _0  + A _1^\intercal   B _1$. This 
	proves the decodability.


	However, as we will explain in the general coding design, directly computing this inverse problem using the classical inversion algorithm might be expensive in some more general cases. Quite interestingly, because of the algebraic structure we designed for the computation strategy (i.e., equation (\ref{eq:bexp})), the decoding process can be viewed as a polynomial interpolation problem (or equivalently, decoding a Reed-Solomon code). 
	
	Specifically, in this example each worker $i$ returns  
	\begin{align}
    	\tilde{C}_i&= \tilde{A} ^\intercal_i \tilde{B}_i = A _0^\intercal  B _1+i (A _0^\intercal  B _0+  A _1^\intercal  B _1) +i^2  A _1^\intercal  B _0,
    	\end{align}
	which is essentially the value of the following polynomial at point $x=i$:
	\begin{align}
    h(x)\triangleq \tilde{A} ^\intercal_i \tilde{B}_i = A _0^\intercal  B _1+x (A _0^\intercal  B _0+  A _1^\intercal  B _1) +x^2  A _1^\intercal  B _0.
    	\end{align}
	Hence, recovering $ C $ using computation results from $3$ workers is equivalent to recovering the linear term coefficient of a quadratic function given its values at $3$ points.
	Later in this section, we will show that by mapping the decoding process to polynomial interpolation, we can achieve almost-linear decoding complexity even for arbitrary parameter values.
	

	\subsection{General Coding Design}\label{sec:general_sub}
	
	Now we present the entangled polynomial code, which achieves a recovery threshold $pmn+p-1$ for any $p$, $m$, $n$ and $N$ as stated in Theorem \ref{th:scheme}.\footnote{For $N<pmn+p-1$, a recovery threshold of $N$ is achievable by definition. Hence we focus on the case where $N\geq pmn+p-1$.} 
	First of all, we evenly divide each input matrix into $pm$ and $pn$ submatrices according to equations (\ref{eq:a}) and (\ref{eq:b}).
	We then assign each worker $i\in\{0,1,...,N-1\}$ an element in $\mathbb{F}$, denoted by $x_i$, and make sure that all $x_i$'s are distinct. 
	Under this setting, we define the following class of computation strategies.
	
		\begin{definition}
	Given parameters $\alpha,\beta,\theta\in\mathbb{N}$, we define the $(\alpha,\beta,\theta)$-polynomial code as
		\begin{align}
		\tilde{A}_i&=\sum_{j=0}^{p-1}\sum_{k=0}^{m-1}  A _{j,k} x_i^{j\alpha+k\beta},\nonumber\\ 
	    \tilde{B}_i&=\sum_{j=0}^{p-1}\sum_{k=0}^{n-1}  B _{j,k} x_i^{(p-1-j)\alpha+k\theta},  \ \ \  \forall\ i\in\{0,1,...,N-1\}.
	\end{align}
		\end{definition}
	In an $(\alpha,\beta,\theta)$-polynomial code, each worker essentially evaluates a polynomial whose coefficients are fixed linear combinations of the products $ A _{j,k}^\intercal  B _{j',k'}$.
		Specifically, each worker $i$ returns
		\begin{align}
		\tilde{C}_i&= \tilde{A} ^\intercal_i \tilde{B}_i \nonumber\\
		&=\sum_{j=0}^{p-1}\sum_{k=0}^{m-1} \sum_{j'=0}^{p-1}\sum_{k'=0}^{n-1} A _{j,k}^\intercal  B _{j',k'} x_i^{(p-1+j-j')\alpha+k\beta+k'\theta}. \label{eq:polyprod}
	\end{align}
		Consequently, when the master receives results from enough workers, it can recover all these linear combinations using polynomial interpolation. Recall that we aim to recover
			\begin{align}
	   C = 
 \begin{bmatrix}
  C_{0,0} & C_{0,1} & \cdots & C_{0,n-1} \\
  C_{1,0} & C_{1,1} & \cdots & C_{1,n-1} \\
  \vdots  & \vdots  & \ddots & \vdots  \\
  C_{m-1,0} & C_{m-1,1} & \cdots & C_{m-1,n-1} 
 \end{bmatrix},
	\end{align}
where each submatrix $C_{k,k'}\triangleq\sum_{j=0}^{p-1} A _{j,k}^\intercal  B _{j,k'}$ is also a fixed linear combination of these products. We design the values of parameters $(\alpha,\beta,\theta)$ such that all these linear combinations appear in (\ref{eq:polyprod}) separately as coefficients of terms of different degrees.  Furthermore, we want to minimize the degree of the polynomial $\tilde{C}_i$, in order to reduce the recovery threshold. 

	One design satisfying these properties is $(\alpha,\beta,\theta)=(1,p,pm)$, 
	i.e,  
	\begin{align}
		\tilde{A}_i&=\sum_{j=0}^{p-1}\sum_{k=0}^{m-1}  A _{j,k} x_i^{j+kp},\nonumber\\
	    \tilde{B}_i&=\sum_{j=0}^{p-1}\sum_{k=0}^{n-1}  B _{j,k} x_i^{p-1-j+kpm}.
	\end{align} 
	Hence, each worker returns the value of the following degree $pmn+p-2$ polynomial at point $x=x_i$:
		\begin{align}
	    h_i(x)&\triangleq\tilde{A} ^\intercal_i \tilde{B}_i \nonumber\\&=\sum_{j=0}^{p-1}\sum_{k=0}^{m-1} \sum_{j'=0}^{p-1}\sum_{k'=0}^{n-1} A _{j,k}^\intercal  B _{j',k'} x_i^{(p-1+j-j')+kp+k'pm},
	\end{align}
	where each $C_{k,k'}$ is exactly the coefficient of the $(p-1+kp+k'pm)$-th degree term.
	Since all $x_i$'s are selected to be distinct, recovering $ C $ given results from any $pmn+p-1$ workers is essentially interpolating $h(x)$ using $pmn+p-1$ distinct points. Because the degree of $h(x)$ is $pmn+p-2$, the output $ C $ can always be uniquely decoded.

	\subsection{Computational complexities}
	\label{sebsec:e}
      

      In terms of complexity, the decoding process of entangled polynomial code can be viewed as interpolating a degree $pmn+p-2$ polynomial for $\frac{rt}{mn}$ times. It is well known that  polynomial interpolation of degree $k$ has a complexity of $O(k\log^2 k \log\log k)$ \cite{von2013modern}.\footnote{When the base field supports FFT, this complexity bound can be improved to $O(k\log^2 k)$. } Therefore, decoding entangled polynomial code only requires at most a complexity of $O(  prt \log^2 (pmn) \log\log (pmn))$, which is almost linear to the input size of the decoder ($\Theta(prt)$ elements). This complexity can be reduced by simply swapping in any faster polynomial interpolation algorithm or Reed-Solomon decoding algorithm.
	In addition, this decoding complexity can also be further improved by exploiting the fact that only a subset of the coefficients are needed for recovering the output matrix.
	
	   	\Qian{Note that given the presented computation framework, each worker is assigned to multiply two coded matrices with sizes of $\frac{r}{m}\times\frac{s}{p}$ and $\frac{s}{p}\times\frac{t}{n}$, which requires a complexity of $O(\frac{srt}{pmn})$.\footnote{\Qian{More precisely, the commonly used cubic algorithm achieves a complexity of $\theta{(\frac{srt}{pmn})}$ for the general case. Improved algorithms has been found in certain cases (e.g., \cite{Strassen1969,4567976,Bini1980,Sch_nhage_1981,doi:10.1137/0211020,Coppersmith:1981:ACM:1398510.1382702,Strassen:1986:AST:1382439.1382931,COPPERSMITH1990251, stothers2010complexity,Williams12multiplyingmatrices}), however, all known approaches requires a super-quadratic complexity.}} This complexity is independent of the coding design, indicating that the entangled polynomial code strictly improves other designs without requiring extra computation at the workers.    
      	Recall that the decoding complexity of entangled polynomial code grows linearly with respect to the size of the output matrix. The decoding overhead becomes negligible compared to workers' computational load in practical scenarios where the sizes of coded matrices assigned to the workers are sufficiently large. Moreover, the fast decoding algorithms enabled by the Polynomial coding approach further reduces this overhead, compared to general linear coding designs.}
      	
      \Qian{Entangled polynomial code also enables improved performances for systems where the data has to encoded online. For instance, if the input matrices are broadcast to the workers and are encoded distributedly, the linearity of entangled polynomial code allows for an in-place algorithm, which does not require addition storage or time complexity. Alternatively, if centralized encoding is required, almost-linear-time algorithms can also be developed similar to decoding: at most a complexity of $O((\frac{sr}{pm}\log^2(pm)\log\log(pm)+\frac{st}{pn}\log^2(pn)\log\log(pn))N)$ is required using fast polynomial evaluation, which is almost linear with respect to the output size of the encoder ($\Theta((\frac{sr}{pm}+\frac{st}{pn})N)$ elements).     
      }


	\section{Converses} \label{sec:conv}
	
	In this section, we provide the proof of Theorem \ref{th:li}. We first prove equation (\ref{eq:tight}) by developing a linear algebraic converse. Then we prove inequality (\ref{bound:it}) through an information theoretic lower bound.
	
	\subsection{Maching Converses for Linear Codes}\label{subsec:li}
	
	To prove equation (\ref{eq:tight}), we start by developing a converse bound on recovery threshold for general parameter values, then we specialize it to the settings where $m=1$ or $n=1$. We state this converse bound in the following lemma:
	
	\begin{lemma}\label{lemma:gb}
		For a distributed matrix multiplication problem with parameters $p$, $m$, $n$, and $N$, we have 
		\begin{align}\label{bound:gen}
		 K^*_{\textup{linear}}\geq \min\{N,\ pm+pn-1\}.
		\end{align}
	\end{lemma}
	When $m=1$ or $n=1$, the RHS of inequality (\ref{bound:gen}) is exactly $K_{\textup{entangled-poly}}$. Hence equation (\ref{eq:tight}) directly follows from Lemma \ref{lemma:gb}. So it only suffices to prove  Lemma \ref{lemma:gb}, and we prove it as follows:
	
	\begin{proof}
		To prove Lemma \ref{lemma:gb}, we only need to consider the following two scenarios:
		
	    (1) If $K^*_{\textup{linear}}=N$, then (\ref{bound:gen}) is trivial.

	    (2) If $K^*_{\textup{linear}}<N$, then we essentially need to show that for any parameter values $p$, $m$, $n$, and $N$ \Qian{satisfying this condition}, we have $K^*_{\textup{linear}}\geq pm+pn-1$. 
	By definition, if such a linear recovery threshold is achievable, we can find a computation strategy, i.e., tensors $a$, $b$, and a class of decoding functions $\boldsymbol{d}\triangleq\{d_{\mathcal{K}}\}$, such that 
		\begin{align}
    	&d_{\mathcal{K}}\left(\left\{\left(\sum_{j',k'} A_{j',k'}^\intercal a_{ij'k'}\right) 
	    \left(\sum_{j'',k''} B_{j'',k''} b_{ij''k''}\right)\right\}_{i\in\mathcal{K}}\right)
	    \nonumber\\ & \ \ \ \  \ \ \ \ \ \ \ \  \ \ \ \ \ \ \ \  \ \ \ \ \ \ \ \  \ \ \ \ \ \ \ \  \ \ \ \ \ \ \ \  \ \ \ \ \  \ \ \ \ = A^\intercal B 
	\end{align}
	for any input matrices $A$ and $B$, and for any subset $\mathcal{K}$ of $K^*_{\textup{linear}}$ workers.
	
	We choose the values of $A$ and $B$, such that each $A_{j,k}$ and $B_{j,k}$ satisfies 
	\begin{align}
	    A_{j,k}&=\alpha_{jk} A_{\textup{c}}, \label{eq:alpha}\\
	    B_{j,k}&=\beta_{jk} B_{\textup{c}}, \label{eq:beta}
	\end{align}
	for some matrices $\alpha\in\mathbb{F}^{p\times m}$, $\beta\in\mathbb{F}^{p\times n}$, and constants $ A_{\textup{c}}\in \mathbb{F}^{\frac{s}{p}\times \frac{r}{m}}$, $ B_{\textup{c}}\in \mathbb{F}^{\frac{s}{p}\times \frac{t}{n}}$ satisfying $A_{\textup{c}}^\intercal B_{\textup{c}}\neq 0$. Consequently, we have
		\begin{align}\label{eq:albe}
    	d_{\mathcal{K}}&\left(\left\{\left(\sum_{j',k'} \alpha_{j'k'} a_{ij'k'}\right) 
	    \left(\sum_{j'',k''} \beta_{j''k''} b_{ij''k''}\right) A_{\textup{c}}^\intercal B_{\textup{c}}\right\}_{i\in\mathcal{K}}\right)
	     \nonumber\\ & \ \ \ \  \ \ \ \ \ \ \ \  \ \ \ \ \ \ \ \  \ \ \ \ \ \ \ \  \ \ \ \ \ \ \ \  \ \ \ \ \ \ \ \  \ \ \ \ \  \ \ \ \ = A^\intercal B  
	\end{align}
	for all possible values of $\alpha$, $\beta$, and $\mathcal{K}$.

	Fixing the value $i$, we can view each subtensor $a_{ijk}$ as a vector of length $pm$, and each  subtensor $b_{ijk}$ as a vector of length $pn$. For brevity, we denote each such vector by $\boldsymbol{a}_i$ and $\boldsymbol{b}_i$ respectively. Similarly, we can also view matrices $\alpha$ and $\beta$ as vectors of length $pm$ and $pn$, and we denote these vectors by $\boldsymbol{\alpha}$ and $\boldsymbol{\beta}$. Furthermore, we can define dot products within these vector spaces following the conventions. 
		Using these notations, (\ref{eq:albe}) can be written as
			\begin{align}\label{eq:salbe}
    	d_{\mathcal{K}}\left(\left\{\left(\boldsymbol{\alpha}\cdot\boldsymbol{a}_i\right) 
	    \left(\boldsymbol{\beta}\cdot\boldsymbol{b}_i\right) A_{\textup{c}}^\intercal B_{\textup{c}}\right\}_{i\in\mathcal{K}}\right)
	    &= A^\intercal B.  
	\end{align}
	
	Given the above definitions, we now prove that within each subset $\mathcal{K}$ of size $K^*_{\textup{linear}}$, the vectors $\{\boldsymbol{a}_i\}_{i \in \mathcal{K}}$ span the space $\mathbb{F}^{pm}$. Essentially, we need to prove that for \Qian{any such} given subset $\mathcal{K}$, there does not exist a non-zero $\alpha\in \mathbb{F}^{p\times m}$ such that \Qian{the corresponding vector $\boldsymbol{\alpha}\in \mathbb{F}^{pm}$ satisfies} $\boldsymbol{\alpha}\cdot\boldsymbol{a}_i = 0$ for all $i \in \mathcal{K}$. Assume the opposite that such an $\alpha$ exists, {so that  $\boldsymbol{\alpha}\cdot\boldsymbol{a}_i$ is always $0$,}
	then the LHS of (\ref{eq:salbe}) becomes a fixed value. On the other hand, since $\alpha$ is non-zero, we can always find different values of $\beta$ such that $\alpha^{\intercal}\beta$ is variable. Recalling (\ref{eq:alpha}) and (\ref{eq:beta}), the RHS of (\ref{eq:salbe}) cannot be fixed if $\alpha^{\intercal}\beta$ is variable, which results in a contradiction.
	
	Now we use this conclusion to prove ($\ref{bound:gen}$). For any fixed $\mathcal{K}$ \Qian{ with size $K^*_{\textup{linear}}$}, let $\mathcal{B}$ be a subset of indices \Qian{in $\mathcal{K}$} such that $\{\boldsymbol{a}_i\}_{i\in\mathcal{B}}$ form a basis. Recall that we are considering the case where $K^*_{\textup{linear}}<N$, meaning that we can find a worker $\tilde{k}\not\in \mathcal{K}$. For convenience, we define $\mathcal{K}^{+}=\mathcal{K}\cup\{\tilde{k}\}$, and $\mathcal{K}^{-} \triangleq \mathcal{K}^{+} \backslash \mathcal{B}$. Obviously, $|\mathcal{B}|=pm$, and $|\mathcal{K}^{-}|=|\mathcal{K}^{+}|-|\mathcal{B}|=K^*_{\textup{linear}}+1-pm$. Hence, it suffices to prove that $|\mathcal{K}^{-}|\geq pn$, which only requires that $\{\boldsymbol{b}_i\}_{i\in\mathcal{K}^{-}}$ forms a basis of $\mathbb{F}^{pn}$. Equivalently, we only need to prove that any $\beta\in \mathbb{F}^{p\times n}$ such that \Qian{its vectorized version $\boldsymbol{\beta}\in \mathbb{F}^{pn}$ satisfies }  $\boldsymbol{\beta}\cdot \boldsymbol{b}_i=0$ for any $i\in\mathcal{K}^{-}$ must be zero. \Qian{For brevity, we let $\mathbb{B}$ denotes the subspace that contains all values of $\beta$ satisfying this condition. }
	
    \Qian{
    To prove this statement, we first  construct a list of matrices as follows, denoted by $\{{\alpha}_i\}_{i \in \mathcal{B}}$}. Recall that  $\{\boldsymbol{a}_i\}_{i \in \mathcal{B}}$ forms a basis. We can find a matrix ${\alpha}_i\in\mathbb{F}^{p\times m}$ for each ${i \in \mathcal{B}}$ such that \Qian{their vectorized version $\{\boldsymbol{\alpha}_i\}_{i \in \mathcal{B}}$ satisfies} $\boldsymbol{\alpha}_i\cdot \boldsymbol{a}_{i'}=\delta_{i,i'}$.\footnote{Here $\delta_{i,j}$ denotes the discrete delta function, i.e., $\delta_{i,i}=1$, and $\delta_{i,j}=0$ for $i\neq j$.} From elementary linear algebra, the vectors $\{\boldsymbol{\alpha}_i\}_{i \in \mathcal{B}}$ also form a basis of $\mathbb{F}^{pm}$. Correspondingly, their matrix version  $\{{\alpha}_i\}_{i \in \mathcal{B}}$ form a basis of $\mathbb{F}^{p\times m}$.
	
	For any $k\in\mathcal{B}$, we define $\mathcal{K}_k=\mathcal{K}^{+}\backslash\{k\}$. Note that $|\mathcal{K}_k|=K^*_{\textup{linear}}$, equation (\ref{eq:salbe}) should also hold for $\mathcal{K}_k$ \Qian{instead of $\mathcal{K}$.} 
\Qian{Moreover,} note that if \Qian{we fix} $\alpha = \alpha_k$, then the corresponding LHS of (\ref{eq:salbe}) remains fixed for any \Qian{$\beta\in \mathbb{B}$}. As a result, 
	$A^\intercal B$ must also be fixed. Similar to the above discussion, this requires that the value of   $\alpha_k^{\intercal}\beta$ be fixed. This value has to be $0$ because $\beta=0$ satisfies our stated condition. 
	
	Now we have proved that any $\beta\in\mathbb{B}$ 
	must also satisfy $\alpha_k^{\intercal}\beta=0$ for any $k\in\mathcal{B}$. Because $\{\alpha_k\}_{k\in\mathcal{B}}$ form a basis of $\mathbb{F}^{p\times m}$, such  $\beta$ acting on $\mathbb{F}^{p\times m}$ through matrix product has to be the zero operator, so $\beta = 0$. As mentioned above, this results in $K^*_{\textup{linear}}\geq pm+pn-1$, which completes the proof of Lemma \ref{lemma:gb} and equation (\ref{eq:tight}).

	

	


	\end{proof}

	\begin{remark}\label{remark:nonld}
	Note that in the above proof, we never used the condition that the decoding functions are linear. Hence, the converse does not require the linearity of the decoder. This fact will be used  later in our discussion regarding the fault-tolerant computing in Appendix \ref{app:ftc}.
	\end{remark}



		

\subsection{Information Theoretic Converse for Nonlinear Codes}\label{subsec:it}

 Now  we prove inequality (\ref{bound:it}) through an information theoretic converse bound. Similar to the proof of equation (\ref{eq:tight}), we start by proving a general converse. 
      
      	\begin{lemma}\label{lemma:git}
		For a distributed matrix multiplication problem with parameters $p$, $m$, $n$, and $N$, if the base field $\mathbb{F}$ is finite, we have 
		\begin{align}\label{bound:git}
		 K^*\geq \max\{pm,pn\}.
		\end{align}
	\end{lemma}
      	When $m=1$ or $n=1$, the RHS of inequality (\ref{bound:git}) is greater than $\frac{1}{2}K_{\textup{entangled-poly}}$. Hence inequality (\ref{bound:it}) directly results from Lemma \ref{lemma:git}, which we prove as follows.
      
		\begin{proof}
		Without loss of generality, we assume $m \geq n$, and aim to  prove  $K^*\geq pm$. \Qian{Specifically, we need to show that any computation strategy has a recovery threshold of at least $pm$, for any possible  parameter values. Recall the definition of recovery threshold. It suffices to prove that for any computation strategy $( 
    	{\boldsymbol{f}}
    	,{\boldsymbol{g}}, {\boldsymbol{d}})$} 
		and any subset
		$\mathcal{K}$ of workers, if the master can recover $C$ given results from workers in $\mathcal{K}$ \Qian{(i.e., the decoding function $d_{\mathcal{K}}$ returns $C$ for any possible values of $A$ and $B$)}, then we must have $|\mathcal{K}|\geq pm$.

		Suppose the condition in the above statement holds. Given each input $A$, the workers can compute $\{\tilde{A}_i\}_{i\in\mathcal{K}}$ using the encoding functions.  
		On the other hand, for any \Qian{fixed} possible  value of $B$, the workers can compute $\{\tilde{C}_i\}_{i\in\mathcal{K}}$ based on  $\{\tilde{A}_i\}_{i\in\mathcal{K}}$. 
		Hence, \Qian{let $\tilde{C}_{i,\textup{func}}$ be a function that returns $\tilde{C}_{i}$ given $B$ as input, $\{\tilde{C}_{i,\textup{func}}\}_{i\in\mathcal{K}}$ is completely determined by $\{\tilde{A}_i\}_{i\in\mathcal{K}}$, without requiring additional information on the value of $A$.} 
		If we view $A$ as a random variable, 
				we have the following Markov chain:
		\begin{align}
		    A\rightarrow \{\tilde{A}_i\}_{i\in\mathcal{K}}\rightarrow \{\tilde{C}_{i,\textup{func}}\}_{i\in\mathcal{K}}.
		\end{align}
		
		Because the master can decode $C$ as a function of $\{\tilde{C}_i\}_{i\in\mathcal{K}}$, \Qian{if we define $C_{\textup{func}}$ similarly as a function that returns $C$ given $B$ as input, $C_{\textup{func}}$ is also completely determined by $\{\tilde{C}_{i,\textup{func}}\}_{i\in\mathcal{K}}$, with no direct dependency on any other variables.}
	Consequently, we have the following extended Markov chain  
			\begin{align}
		    A\rightarrow \{\tilde{A}_i\}_{i\in\mathcal{K}}\rightarrow \{\tilde{C}_i\}_{i\in\mathcal{K}}\rightarrow C_{\textup{func}}.
		\end{align}

	\Qian{Note that by definition, $C_{\textup{func}}$ has to satisfy $C_{\textup{func}}(B)=A^\intercal B$ for any $A\in\mathbb{F}^{s\times r}$ and $B\in\mathbb{F}^{s\times t}$. Hence, $C_{\textup{func}}$ is essentially a linear
	operator uniquely determined by $A$, defined as multiplication by $A^{\intercal}$. Conversely, one can show that distinct values of $A$ leads to distinct operators, which directly follows from the definition of matrix multiplication.  }
		Therefore, the input matrix $A$ can be exactly determined from $C_{\textup{func}}$, i.e., $\mathrm{H}(A | C_{\textup{func}})=0$. 
				Using the data processing inequality, we have $\mathrm{H}(A |\{\tilde{A}_i\}_{i\in\mathcal{K}})=0$.
				
				Now let $A$ be uniformly randomly sampled from $\mathbb{F}^{s\times r}$, and we have $\mathrm{H}(A)=sr\log_2 |\mathbb{F}|$ bits. On the other hand, each $\tilde{A}_i$ consists of $\frac{sr}{pm}$ elements, which has an entropy of at most $\frac{sr}{pm}\log_2 |\mathbb{F}|$ bits. 
		Consequently, we have 
		\begin{align}
		    |\mathcal{K}|\geq \frac{\mathrm{H}(A)}{\max\limits_{i\in\mathcal{K}} \mathrm{H}(\tilde{A}_i)} \geq pm.
		\end{align}
		This concludes the proof of Lemma \ref{lemma:git} and inequality (\ref{bound:it}).
	\end{proof}

    \section{Factor of $2$ characterization of Optimum Linear Recovery Threshold}\label{sec:2}



      In this section, we provide the proof of Theorem \ref{th:2}. \Qian{Specifically, we need to provide a computation strategy that achieves a recovery threshold of at most $2R(p,m,n)-1$  for all possible values of $p$, $m$, $n$, and $N$, as well as a converse result showing that any linear computation strategy requires at least $N\geq R(p,m,n)$ workers for any $p$, $m$, and $n$.}
      

    
		
		\Qian{The proof is accomplished in $2$ steps. In Step $1$, we show that any linear code for matrix multiplication is equivalently an upper bound construction of the bilinear complexity $R(p,m,n)$, and vice versa. This result indicates the equality between $R(p,m,n)$ and the minimum required number of workers, which proves the needed converse. It also converts any matrix multiplication into the 
		computation of element-wise products given two vectors of length $R(p,m,n)$. 
		Then in Step $2$, we show that we can find an optimal computation strategy for this augmented computing task. We develop a variation of the entangled polynomial code, which achieves a recovery threshold of $ 2R(p,m,n)-1$.}

		
		 
      \Qian{For Step $1$, we first formally define upper bound constructions for bilinear complexity. 
      \begin{definition}
      Given parameters $p$, $m$, $n$, an \emph{upper bound construction for bilinear complexity} $R(p,m,n)$ with \emph{rank} $R$ is a tuple of tensors $a\in\mathbb{F}^{R\times p\times m}$, $b\in \mathbb{F}^{R\times p\times n}$, and $c\in \mathbb{F}^{R\times m\times n}$ such that for any matrices $A\in\mathbb{F}^{p\times m}$, $B\in \mathbb{F}^{p\times n}$,
      \begin{align}
    	\sum_{i}c_{ijk}\left(\sum_{j',k'} A_{j'k'} a_{ij'k'}\right) 
	    &\left(\sum_{j'',k''} B_{j''k''} b_{ij''k''}\right)\nonumber\\
	    &
	    = \sum_{\ell} A_{\ell j} B_{\ell k}.
	\end{align}
      \end{definition}}
      \Qian{Recall the definition of linear codes. One can verify that any upper bound construction with rank $R$ is equivalently a linear computing design using $R$ workers when the sizes of input matrices are given by $A\in\mathbb{F}^{p\times m}$, $B\in \mathbb{F}^{p\times n}$. Note that matrix multiplication follows the same rules for any block matrices, this equivalence holds true for any input sizes.\footnote{Rigorously, it also requires the linear independence of the $A_i^\intercal B_j$'s, which can be easily proved.} Specifically, given an upper bound construction $(a,b,c)$ with rank $R$, and for general inputs $A\in\mathbb{F}^{s\times r}$, $B\in \mathbb{F}^{s\times t}$,} any block of the final output $C$ can be computed as 
        \begin{align}\label{eq:bi1}
	    C_{j,k}=\sum_{i}c_{ijk}\tilde{A}^\intercal_{i,\textup{vec}}\tilde{B}_{i,\textup{vec}},
	\end{align}
	where $\tilde{A}_{i,\textup{vec}}$ and $\tilde{B}_{i,\textup{vec}}$ \Qian{are linearly encoded matrices stored by $R$ workers, defined as}
	\begin{align}\label{eq:bi2}
    	\tilde{A}_{i,\textup{vec}}\triangleq\sum_{j,k} A_{j,k} a_{ijk},  \ \ \ \ 
	    \tilde{B}_{i,\textup{vec}}\triangleq\sum_{j,k} B_{j,k} b_{ijk}.
	\end{align}
	\Qian{Conversely, one can also show that any linear code using $N$ workers is equivalently an upper bound construction with rank $N$. This equivalence relationship provides a one-to-one mapping between linear codes and upper bound constructions. }
      
     \Qian{Recall the definition of bilinear complexity (provided in Section \ref{sec:res}), which essentially states that the minimum achievable rank $R$ equals $R(p,m,n)$. We have shown that the minimum number of workers required for any linear code is given by the same quantity, which proves the coverse.  }
     \Qian{In terms of achievability, we have also proved the existence of a linear computing design using $R(p,m,n)$ workers, where the encoding and decoding are characterized by some tensors $a\in\mathbb{F}^{R(p,m,n)\times p\times m}$, $b\in \mathbb{F}^{R(p,m,n)\times p\times n}$, and $c\in \mathbb{F}^{R(p,m,n)\times m\times n}$ satisfying equation (\ref{eq:bidef}), following equations (\ref{eq:bi1}) and (\ref{eq:bi2}). } 
      This achievability scheme essentially converts matrix multiplication into a problem of computing the element-wise product of two ``vectors'' $\tilde{A}_{i,\textup{vec}}$ and $\tilde{B}_{i,\textup{vec}}$, each of length $R(p,m,n)$. Specifically, the master only needs $\tilde{A}^\intercal_{i,\textup{vec}}\tilde{B}_{i,\textup{vec}}$ for decoding the final output.

		Now in Step $2$, we develop the optimal computation strategy for this augmented computation task. Given two arbitrary vectors $\tilde{A}_{i,\textup{vec}}$ and $\tilde{B}_{i,\textup{vec}}$ of length $R(p,m,n)$, we want to achieve a recovery threshold of $2R(p,m,n)-1$ for computing their element-wise product using $N$ workers, each of which can multiply two coded vectors of length $1$. 
    \Qian{As we have explained in Section \ref{sec:general_sub}, a recovery threshold of $N$ is always achievable, so we only need to focus on the scenario where $N\geq 2R(p,m,n)-1$.}

		The main coding idea is to first view the elements in each vector as values of a degree $R(p,m,n)-1$ polynomial at $R(p,m,n)$ different points. Specifically, given $R(p,m,n)$ distinct elements in the field $\mathbb{F}$, denoted by $x_0, x_1, \dots, x_{R(p,m,n)-1}$, we find polynomials $\tilde{f}$ and $\tilde{g}$ of degree $R(p,m,n)-1$,  whose coefficients are matrices, such that
			\begin{align}
	    \tilde{f}(x_i)&=\tilde{A}_{i,\textup{vec}}\\
	    \tilde{g}(x_i)&=\tilde{B}_{i,\textup{vec}}.
	\end{align}
	Recall that we want to recover $\tilde{A}_{i,\textup{vec}}^\intercal\tilde{B}_{i,\textup{vec}}$, which is essentially recovering the values of the degree $2R(p,m,n)-2$ polynomial $\tilde{h}\triangleq\tilde{f}^\intercal\tilde{g}$ at these $R(p,m,n)$ points. 
	Earlier in this paper, we already developed a coding structure that allows us to recover polynomials of this form. We now reuse the idea in this construction. 
	
	Let $y_0$, $y_1$, ..., $y_{N-1}$ be distinct elements of $\mathbb{F}$. We let each worker $i$ store	\begin{align}
    	\tilde{A}_{i}=\tilde{f}(y_i), \\ 
	    \tilde{B}_{i}=\tilde{g}(y_i),
	\end{align}
	which are linear combinations of the input submatrices. More Specifically, 
	\begin{align}
    	\tilde{A}_{i}&=\sum_j \tilde{A}_{j,\textup{vec}} \cdot \prod_{k\neq {j}}\frac{(y_i-x_k)}{(x_j-x_k)}, \\ 
	    \tilde{B}_{i}&=\sum_j \tilde{B}_{j,\textup{vec}} \cdot  \prod_{k\neq {j}}\frac{(y_i-x_k)}{(x_j-x_k)}.
	\end{align}
	
	After computing the product, each worker essentially evaluates the polynomial $\tilde{h}$ at $y_i$. Hence, from the results of any $2R(p,m,n)-1$ workers, the master can recover $\tilde{h}$, which has degree $2R(p,m,n)-2$, and proceed with decoding the output matrix $C$. This construction achieves a recovery threshold of $2R(p,m,n)-1$, which proves the upper bound in Theorem \ref{th:2}.
	
			\begin{remark}
			The computation strategy we developed in Step $2$ provides a tight upper bound on the characterization of the optimum linear recovery threshold for computing element-wise product of two arbitrary vectors using $N$ machines. Its optimality naturally follows from Theorem \ref{th:li}, given that the element-wise product of two vectors contains all the information needed to compute the dot-product, which is a special case of matrix multiplication. We formally state this result in the following corollary.
			\end{remark}
	  \begin{corollary}
	  Consider the problem of computing the element-wise product of two vectors of length $R$ using $N$ workers, each of which can store a linearly coded element of each vector and return their product to the master. The optimum linear recovery threshold, denoted as $K^*_{\textup{e-prod-linear}}$, is given by the following equation:\footnote{Obviously, we need $N\geq R$ to guarantee the existence of a valid computation strategy.}
	  \begin{align}
	      K^*_{\textup{e-prod-linear}}=\min\{N,2R-1\}.
	  \end{align}
      \end{corollary}

		\begin{remark}
      Note that Step $2$ of this proof does not require the computation strategy to be linear. 
      Hence, using exactly the same coding approach, we can easily extend this result to non-linear codes, and prove a similar factor-of-$2$ characterization for the optimum recovery threshold $K^*$, formally stated in the following corollary.
      \end{remark}
      
      \begin{corollary}
      For a distributed matrix multiplication problem with parameters $p$, $m$, and $n$, let $N^*(p,m,n)$ denotes the minimum number of workers such that a valid (possibly non-linear) computation strategy exists. Then for \emph{all} possible values of $N$, we have 
       \begin{align}
	     N^*(p,m,n) \leq K^*\leq 2N^*(p,m,n)-1.
	  \end{align}
      \end{corollary}

      	\begin{remark}
      \Qian{
      Finally, note that the computing design provided in this section can be applied any upper bound construction with rank $R$, achieving a recovery threshold of $2R-1$,
      its significance is two-fold. Using constructions that achieves bilinear complexity, it proves the existence of a factor-of-$2$ optimal computing scheme, which achieves the same recovery threshold while tolerating arbitrarily many stragglers. On the other hand, for cases where $R(p,m,n)$ is not yet known, explicit coding constructions can still be obtained (e.g., using the well know Strassen's result \cite{Strassen1969}, as well as any other known constructions, such as ones presented in \cite{4567976,doi:10.1137/0120004,laderman1976noncommutative,Bini1980,Sch_nhage_1981,doi:10.1137/0211020,Coppersmith:1981:ACM:1398510.1382702,Strassen:1986:AST:1382439.1382931,COPPERSMITH1990251, stothers2010complexity,Drevet2011,Williams12multiplyingmatrices,Smirnov2013,sedoglavic2017non,sedoglavic:hal-01572046}), which enables further improvements upon the basic entangled polynomial code.}
      \end{remark}
      
      \subsection{Computational complexities}
      
      \Qian{Algorithmically, decoding the improved version of entangled polynomial code can be completed in two steps. In step $1$, the master can first recover the element-wise products $\{\tilde{A}_{i,\textup{vec}}^\intercal\tilde{B}_{i,\textup{vec}}\}_{i=1}^{R(p,m,n)}$, by Lagrange-interpolating a degree $2R(p,m,n)-1$ polynomial at $R(p,m,n)$ points, for $\frac{rt}{mn}$ times. Similar to the entangled polynomial code, it requires a complexity of at most $O(  \frac{rt}{mn} R(p,m,n) \log^2 (R(p,m,n)) \log\log (R(p,m,n)))$, which is almost linear to the input size of the decoder ($\Theta(\frac{rt}{mn}R(p,m,n))$ elements). Then in Step $2$, the master can recover the final results by linearly combining these products, following equation 
      (\ref{eq:bi1}). Note that without even exploiting any algebraic properties of the tensor construction, the natural computing approach achieves a complexity of $\Theta({rt}R(p,m,n))$ for computing the second step. 
      This already achieves a strictly smaller decoding complexity compared with a general linear computing design, which could requires inverting an $R(p,m,n)$-by-$R(p,m,n)$ matrix.\footnote{Similar to matrix multiplication, inverting a $k$-by-$k$ matrix requires a complexity of $O(k^3)$. Faster algorithms has been developed, however, all known results requires super-quadratic complexity.   }}
      
     \Qian{ Moreover, note that most commonly used upper bound constructions are based on the sub-multiplicativity of $R(p,m,n)$, further improved decoding algorithms can be designed when these constructions are used instead. As an example, consider Strassen's construction, which achieves a rank of $R=7^k\geq R(2^k,2^k,2^k)$. The final outputs can essentially be recovered 
     given the intermediate products $\{\tilde{A}_{i,\textup{vec}}^\intercal\tilde{B}_{i,\textup{vec}}\}_{i=1}^{R}$ by following the last few iterations of Strassen's Algorithm, requiring only a linear complexity $\Theta(\frac{rt}{mn}R)$.  This approach achieves an overall decoding complexity of   $O(  \frac{rt}{mn} R\log^2 R \log\log R)$, which is almost linear to the input size of the decoder.} 
      
	   	\Qian{Similar to the discussion in Section 	\ref{sebsec:e}, the computational complexity at each worker is $O(\frac{srt}{pmn})$, which is independent of the coding design. Hence, the improved version of the entangled polynomial code also does not require extra computation at the workers, and the decoding overhead becomes negligible when sizes of the coded submatrices are sufficiently large. Improved performances can also be obtained for systems that requires online encoding, following similar approaches used in decoding.  } 
	   	

      \section{Concluding Remarks}

      In this paper, we studied the coded distributed matrix multiplication problem and proposed entangled polynomial codes, which allows optimal straggler mitigation and orderwise improves upon the prior arts. Based on our proposed coding idea, we proved a fundamental connection between the optimum linear recovery threshold and the bilinear complexity, which characterizes the optimum linear recovery threshold within a factor of $2$ for all possible parameter values. 
      The techniques developed in this paper can be directly applied to many other problems, including coded convolution and fault-tolerant computing, providing matching characterizations. 
      By directly extending entangled polynomial codes to  secure \cite{chang2018capacity, 8382305, DBLP:journals/corr/abs-1810-13006, DBLP:journals/corr/abs-1812-09962, 8613446, DBLP:journals/corr/abs-1901-07705, kim2019private, 8675905 ,chang2019upload, 8761275, nodehi2019secure, jia2019capacity, kakar2019uplinkdownlink, aliasgari2019private, d2019degree}, private \cite{kim2019private,chang2019upload, aliasgari2019private, 8832193}, and batch \cite{jia2019capacity, jia2019cross, jia2019generalized} distributed matrix multiplication, we can also unboundedly improve all other block-partitioning based schemes 
 \cite{DBLP:journals/corr/abs-1901-07705, aliasgari2019private, 8613446, jia2019cross, jia2019generalized}, achieving subcubic recovery threshold while enabling flexible resource tradeoffs.\footnote{For details, see \cite{yu2020entangled}.} Entangled polynomial codes has also inspired recent development of coded computing schemes for general polynomial computations \cite{pmlr-v89-yu19b}, secure/private computing \cite{so2019codedprivateml}, and secure sharding in blockchain systems \cite{li2018polyshard}.

       One interesting follow-up direction is to find better characterization of the optimum linear recovery threshold.  Although this problem is completely solved for cases including $m=1$, $n=1$, or $p=1$, there is room for improvement in general cases.
            Another interesting question is whether there exist non-linear coding strategies that strictly out-perform linear codes, especially for the important case where the input matrices are large ($s,r,t \gg p,m,n$), while allowing for efficient decoding algorithms with almost linear complexity. 
      Finally, the main focus of this paper is to provide optimal algorithmic solutions for matrix multiplication on general fields. Although, when the base field is infinite, one can instead embed the computation into finite fields to avoid practical issues such as numerical error and computation overheads (see discussions in \cite{NIPS2017_7027, 8849245}). It is an interesting following direction to find new quantization and computation schemes to study optimal tradeoffs between these measures.

 \appendices


\section{The Optimum Linear Recovery Threshold for Coded Convolution}\label{app:conv}

In this appendix, we first provide the problem formulation for coded convolution, then we prove Theorem \ref{thm:conv}, which shows the optimality of Polynomial Code for Coded Convolution. 

\subsection{System Model and Problem Formulation}

Consider a convolution task with two input vectors 
		\begin{align}
	    \boldsymbol{a}=[\boldsymbol{a}_0~\boldsymbol{a}_1~...~\boldsymbol{a}_{m-1}],\ \ \ \ \ \ \ \  
	    \boldsymbol{b}=[\boldsymbol{b}_0~\boldsymbol{b}_1~...~\boldsymbol{b}_{n-1}],
	\end{align}
	where all $\boldsymbol{a}_i$'s and $\boldsymbol{b}_i$'s are vectors of length $s$ over a sufficiently large field $\mathbb{F}$.
	We want to compute  $\boldsymbol{c}\triangleq\boldsymbol{a} * \boldsymbol{b}$	
	 using a master and $N$ workers. Each worker can store two vectors of length $s$, which are functions of $\boldsymbol{a}$ and $\boldsymbol{b}$ respectively. We refer to these functions as the \emph{encoding functions}, denoted by $	({\boldsymbol{f}}
    	,{\boldsymbol{g}})$ similar to the matrix multiplication problem.

	 Each worker computes the convolution of its stored vectors, and returns it to the master. The master only waits for the fastest subset of workers, before proceeding to decode $\boldsymbol{c}$. 
	 Similar to the matrix multiplication problem, we define the \emph{recovery threshold} given the \emph{encoding functions}, denoted by $K( {\boldsymbol{f}} ,{\boldsymbol{g}})$, as the minimum number of workers that the master needs to wait that guarantees the existence of valid decoding functions. We aim to characterize the optimum recovery threshold achievable by any linear encoding functions, denoted by $K^*_{\textup{conv-linear}}$, and identify an optimal computation strategy that achieves this optimum threshold. 
	 
	 \subsection{Proof of Theorem \ref{thm:conv}}
	 
	 Now we prove Theorem \ref{thm:conv}, which completely solves the above problem. As we have shown in \cite{NIPS2017_7027}, the recovery threshold stated in Theorem \ref{thm:conv} is achievable using a variation of polynomial code. This result proves an upperbound of $K^*_{\textup{conv-linear}}$. It also identifies an optimal computation strategy. Hence, in this section we focus on proving the matching converse.
	 
	 Specifically, we aim to prove that given any problem parameters $m$, $n$, and $N$, for any computation strategy, if the encoding functions $	({\boldsymbol{f}},{\boldsymbol{g}})$ are linear, then its recovery threshold is at least $m+n-1$. We prove it by contradiction. 
	 
	 Assume the opposite, then the master can recover $\boldsymbol{c}$ using results from a subset of at most $m+n-2$ workers. We denote this subset by $\mathcal{K}$.
	 Obviously, we can find a partition of $\mathcal{K}$ into two subsets, denoted by $\mathcal{K}_{\textup{a}}$ and $\mathcal{K}_{\textup{b}}$, such that $|\mathcal{K}_{\textup{a}}|\leq m-1$ and $|\mathcal{K}_{\textup{b}}|\leq n-1$. Note that the encoding functions of workers in  $\mathcal{K}_{\textup{a}}$ collaboratively and linearly maps $\mathbb{F}^{ms}$ to $\mathbb{F}^{(m-1)s}$, which has a non-zero kernel. Hence, we can find a non-zero input vector $\boldsymbol{a}$ such that all workers in $\mathcal{K}_{\textup{a}}$ returns $0$. Similarly, we can find a non-zero $\boldsymbol{b}$ such that all workers in $\mathcal{K}_{\textup{b}}$ returns $0$. Recall that $ \mathcal{K}_{\textup{a}}\cup \mathcal{K}_{\textup{b}}=\mathcal{K}$. Consequently, when the master receives $0$ from all workers in $\mathcal{K}$, the decoding function returns $\boldsymbol{a} * \boldsymbol{b}$.  
	 
	 This convolution product must be the $\boldsymbol{0}$ vector, given that the workers return the same results under zero inputs. However, note that the convolution operator has no zero-divisor. Either $\boldsymbol{a}$ or $\boldsymbol{b}$ has to be zero, which contradicts the non-zero assumptions. Hence, we have $K( {\boldsymbol{f}} ,{\boldsymbol{g}})\geq m+n-1$. This concludes the proof of Theorem \ref{thm:conv}.

\section{An Equivalence Between Fault Tolerance and Straggler Mitigation}\label{app:ftc}

In this appendix, we start by formulating a fault-tolerant computing problem for matrix multiplication, then we prove Theorem \ref{thm:ftc} by building a connection between straggler mitigation and fault tolerance, by extending the concept of Hamming distance to coded computing.

\subsection{Problem Formulation}

We consider a matrix multiplication problem with two input matrices $ A \in \mathbb{F}^{s\times r}$ and $ B \in \mathbb{F}^{s\times t}$, and we are interested in computing $ C \triangleq A ^\intercal B $ using a master node and $N$ worker nodes, where
		each worker can store $\frac{1}{pm}$ fraction of $A$ and $\frac{1}{pn}$ fraction of $B$.
		Similar to the straggler mitigation problem, each worker $i$ can store two coded matrices $\tilde{A}_i\in\mathbb{F}^{\frac{s}{p}\times \frac{r}{m}}$ and $\tilde{B}_i\in\mathbb{F}^{\frac{s}{p}\times \frac{t}{n}}$, computed based on 
		$ A $ and $ B $ respectively. 
		Each worker can compute the product  $\tilde{C}_i\triangleq\tilde{A}_i^\intercal\tilde{B}_i$, and return it to the master. 
	    Unlike the straggler setting, the master waits for all workers before proceeding to recover the final output $C$. However, a subset of workers can return error results, and the master has no information on which subset of results are false.
	    Under this setting, the master wants to: (1) determine if there is an error in the workers' outputs, and (2) try to recover the final output $C$ using the possibly false computing results from the workers.

		Given the above system model, we formulate this fault-tolerant computing problem based on the following terminology. Similar to our main problem in this paper, we define the \emph{encoding functions} and denote them by $({\boldsymbol{f}},{\boldsymbol{g}})$. We also define the \emph{decoding function} for the master, however in this problem it can either return an estimate of $C$, or report an error. We only consider the \emph{valid} decoding functions, which always correctly decodes $C$ when no worker is making mistakes. 
		
		For any integer $E$, we say the encoding functions can \emph{detect} $E$ errors if we can find a decoding function that
		either returns the correct value of $C$ or reports an error, when no more than $E$ workers are making mistakes. 
		Moreover, we say the encoding functions can \emph{correct} $E$ errors, if the decoding function always correctly decodes $C$. 
	    We denote the maximum possible integer $E$ given these two criteria by $E_{\textup{detect}}({\boldsymbol{f}},{\boldsymbol{g}})$ and $E_{\textup{correct}}({\boldsymbol{f}},{\boldsymbol{g}})$ respectively. 
		
			We aim to find encoding functions that allows detecting and correcting the maximum possible number of errors. Among all possible computation strategies, we are particularly interested in \emph{linear encoding functions}, as defined in Section \ref{sec:sys}. Given the above terminology, we define the following concepts.
    		\begin{definition}
    	For a distributed matrix multiplication problem of computing $A ^\intercal B $ using $N$ workers, we define the \emph{maximum detectable errors} and the \emph{maximum detectable errors}, denoted by $E^*_{\textup{detect}}$ and $E^*_{\textup{correct}}$ respectively, as the maximum possible values of $E_{\textup{detect}}({\boldsymbol{f}},{\boldsymbol{g}})$ and $E_{\textup{correct}}({\boldsymbol{f}},{\boldsymbol{g}})$ over the set of all encoding functions that are linear.
    	\end{definition}
    	
     	Our goal is to characterize the values of $E^*_{\textup{detect}}$ and $E^*_{\textup{correct}}$,  and to find optimal computation strategies to achieve these values. We are also interested in extending these characterizations to non-linear codes.

		\subsection{Proof of Theorem \ref{thm:ftc}}
		
		We start by defining some concepts, which allows connecting the fault-tolerant computing problem to the straggler mitigation problem.
			\begin{definition}
        We define the \emph{Hamming distance} of any encoding functions $({\boldsymbol{f}},{\boldsymbol{g}})$, denoted by $d({\boldsymbol{f}},{\boldsymbol{g}})$, as the maximum integer $d$ such that for any two pairs of input matrices whose products $C$ are different, at least $d$ workers compute different values of $ \tilde{C}_i$.   
    	\end{definition}
    			\begin{definition}
        We define the \emph{Recovery threshold} of any encoding functions $({\boldsymbol{f}},{\boldsymbol{g}})$, denoted by $K({\boldsymbol{f}},{\boldsymbol{g}})$, as the minimum possible recovery threshold given any decoding functions.
    	\end{definition}
		
        We prove that all these three mentioned criteria for designing encoding functions are directly connected by the Hamming distance, which is formally stated as follows.
        
        \begin{lemma}\label{lemma:hamming}
        For any (possibly non-linear) computation strategy,  we have
        \begin{align}
            K({\boldsymbol{f}},{\boldsymbol{g}})&=N-d({\boldsymbol{f}},{\boldsymbol{g}})+1,\label{eq:rec}\\
            E_{\textup{detect}}({\boldsymbol{f}},{\boldsymbol{g}})&=d({\boldsymbol{f}},{\boldsymbol{g}})-1,\label{eq:det}\\
            E_{\textup{correct}}({\boldsymbol{f}},{\boldsymbol{g}})&=\left\lfloor\frac{d({\boldsymbol{f}},{\boldsymbol{g}})-1}{2}\right\rfloor\label{eq:cor}.
        \end{align}
        \end{lemma}

\begin{remark}
Lemma \ref{lemma:hamming} essentially indicates that optimizing the straggler mitigation performance over any class of \Qian{en}coding designs is equivalently optimizing its performance in the fault tolerance setting. Furthermore, all these previously mentioned metrics can be simultaneously optimized by the codes with the maximum possible Hamming distance. Hence, there is no tension among these metrics.  This result bridges the rich literature of coding theory and distributed computing. 
\end{remark}

\begin{remark}
\Qian{In terms of achievability, Lemma \ref{lemma:hamming} also provides a large class of coding designs for fault-tolerant computing. Specifically, it indicates that given any computing scheme (e.g., the entangled polynomial code, or its improved version) that achieves a certain recovery threshold, denoted by $K$. Using the same encoding functions, we can obtain a fault-tolerant scheme that detects up to $N-K$ errors, or correct up to $\lfloor\frac{N-K}{2}\rfloor$ errors.    }
\end{remark}
        
        \begin{proof}[Proof of Lemma \ref{lemma:hamming}]
        
 Lemma \ref{lemma:hamming} is a direct consequence of the classical coding theory, given that mitigating straggler effects is essentially correcting erasure errors, and tolerating false results in computing is essentially correcting arbitrary error. Hence, we only provide the proof of (\ref{eq:rec}), where equations (\ref{eq:det}) and (\ref{eq:cor}) can be proved using similar approaches.         
 
 Specifically, we want to prove that for any integer $K$, a recovery threshold of $K$ is achievable by some encoding functions if and only if their Hamming distance is greater or equal to $N-K+1$. If $K$ is achievable, it means that we can find decoding functions that uniquely determines the value of $C$ given results from any $K$ workers. Equivalently, for distinct values of $C$, at least $N-K+1$ workers has to return distinct results. 
  Recall that the recovery threshold is the minimum of such integer $K$, and 
 the Hamming distance is the maximum integer that corresponds to $N-K+1$. We have $K({\boldsymbol{f}},{\boldsymbol{g}})=N-d({\boldsymbol{f}},{\boldsymbol{g}})+1$.
 
        \end{proof} 

Now we continue to prove Theorem \ref{thm:ftc} using  Lemma \ref{lemma:hamming}. As mentioned in Remark \ref{remark:nonld}, the proof of Theorem  \ref{th:li} essentially completely characterizes the optimum recovery threshold over all linear encoding functions for $m=1$ or $n=1$, which is given by $K_{\textup{entangled-poly}}$. Hence, using Lemma \ref{lemma:hamming}, we directly obtain that if $m=1$ or $n=1$, we have
		\begin{align}
		 E^*_{\textup{detect}} &={N-K_{\textup{entangled-poly}}},\\
		 E^*_{\textup{correct}}&=\left\lfloor\frac{N-K_{\textup{entangled-poly}}}{2}\right\rfloor.
		\end{align}
		This concludes the proof of Theorem \ref{thm:ftc}.




	\section*{acknowledgement}
The authors would like to thank Jie Li for careful reading of the manuscript and discussions.

   \bibliographystyle{ieeetr}
   \bibliography{mat}

\end{document}